\documentclass[11pt]{article}

\usepackage{fullpage}
\usepackage{booktabs} 
\usepackage{verbatim}
\usepackage{amsmath, amssymb, amsthm}
\usepackage{tikz}
\usetikzlibrary{arrows}
\usepackage{color, xcolor}
\usepackage{ifthen}
\usepackage[caption=false,font=footnotesize]{subfig}
\usepackage{thmtools, thm-restate}
\usepackage{algorithmic}
\usepackage[boxed]{algorithm}

\usepackage{hyperref}
\hypersetup{colorlinks=true,urlcolor=blue,linkcolor=blue,citecolor=blue}

\newcommand{\compilehidecomments}{true}

\ifthenelse{ \equal{\compilehidecomments}{true} }{%
	\newcommand{\wei}[1]{}
	\newcommand{\shanghua}[1]{}
	\newcommand{\hanrui}[1]{}
}{
	\newcommand{\wei}[1]{{\color{blue!50!black}  [\text{Wei:} #1]}}
	\newcommand{\shanghua}[1]{{\color{brown!60!black} [\text{Shanghua:} #1]}}
	\newcommand{\hanrui}[1]{{\color{green!60!black} [\text{Hanrui:} #1]}}
}

\newcommand{\I}{\mathcal{I}}

\theoremstyle{plain}

\newtheorem{note}{Note}[section]
\newtheorem{axiom}{Axiom}[section]
\newtheorem{mythm}{Theorem}[section]
\newtheorem{mylem}{Lemma}[section]
\newtheorem{lemma}{Lemma}[section]
\newtheorem{myprop}{Proposition}[section]
\newtheorem{mydef}{Definition}[section]

\newcommand{\E}{\mathbb{E}}
\newcommand{\R}{\mathbb{R}}
\newcommand{\bT}{\mathbb{T}}

\renewcommand{\bT}{\boldsymbol{T}}

\newcommand{\bv}{\boldsymbol{v}}

\newcommand{\cL}{\mathcal{L}}

\newcommand{\epsi}{\hat{{\psi}}}

\newcommand{\vv}{\vec{v}}

\newcommand{\vb}{\vec{b}}
\newcommand{\vzero}{\vec{0}}

\newcommand{\Closeness}{{\rm cls}}
\newcommand{\Degree}{{\rm deg}}
\newcommand{\Harmonic}{{\rm har}}
\newcommand{\Reachability}{{\rm rch}}
\newcommand{\SoI}{{\rm SoI}}
\newcommand{\SNI}{{\rm SNI}}
\newcommand{\SNSSoI}{{\rm SNSSoI}}
\newcommand{\Shapley}{{\rm Shapley}}
\newcommand{\Group}{{\rm grp}}
\newcommand{\Individual}{{\rm ind}}

\newcommand{\IndDegree}{\text{ind-deg}}

\def\est{{\rm est}}
\def\LB{{\rm LB}}

\def\INPUT{\REQUIRE}
\def\OUTPUT{\ENSURE}

\def\ICERR{{\sf ICE-RR}}

\begin{document}
	
\title{A Systematic Framework and Characterization of \\ Influence-Based Network Centrality}

\author{Wei Chen \\
  Microsoft Research \\
  \texttt{weic@microsoft.com}
  \and
  Shang-Hua Teng \\
  University of Southern California \\
  \texttt{shanghua@usc.edu}
  \and
  Hanrui Zhang \\
  Duke University \\
  \texttt{hrzhang@cs.duke.edu}
}


\maketitle

\begin{abstract}

In this paper, we present a framework 
  for studying the following fundamental question in network analysis:
\begin{quote}
{\em How should one assess the centralities of nodes
   in an information/influence propagation process over
   a social network?
}
\end{quote}

Our framework systematically extends a family of classical
  graph-theoretical centrality formulations,
  including degree centrality,  harmonic centrality,
  and their ``sphere-of-influence'' 
 generalizations, to influence-based network centralities.
We further extend natural group centralities from 
  graph models to influence models, 
   since group cooperation is essential in social influences.
This in turn enables us to assess
  individuals' centralities in group influence settings
  by applying the concept of Shapley value from
  cooperative game theory.

Mathematically, using the property that these centrality formulations
   are Bayesian\footnote{Meaning that they are linear to the convex
combination of influence instances.},
   we prove the following characterization theorem:
Every influence-based centrality formulation in this family is the 
{\em  unique Bayesian centrality} that conforms with its
  corresponding graph-theoretical centrality formulation.
Moreover, the uniqueness is fully determined by 
  the centrality formulation on the class of layered graphs,
  which  is derived from a beautiful algebraic
  structure of influence instances modeled by cascading sequences.
Our main mathematical result that layered graphs in fact form a 
  basis for the space of influence-cascading-sequence profiles
   could also be useful in 
  other studies of network influences.
We further provide an algorithmic framework for efficient
  approximation of  these influence-based centrality measures.  

Our study provides a systematic
  road map for comparative analyses of different
  influence-based centrality formulations, as well as 
  for transferring graph-theoretical concepts to 
  influence models.
\end{abstract}

%
%


\section{Introduction}\label{sec:Intro}

Network influence is a fundamental subject in network sciences
 \cite{RichardsonDomingos,DomingosRichardson,kempe03,borgatti06}.
It arises from vast real-world backgrounds, ranging 
  from epidemic spreading/control, to viral marketing, 
  to innovation,  and to political campaign.
It also provides a family of concrete 
   and illustrative examples for studying 
   network phenomena --- particularly regarding the interplay 
   between network dynamics and graph structures --- 
   which require solution concepts beyond 
   traditional graph theory \cite{CT17}.
As a result, network influence is a captivating subject for 
  theoretical modeling, mathematical characterization, 
  and algorithmic analysis
  \cite{RichardsonDomingos,DomingosRichardson,kempe03,ChenWeiBook}.

\wei{Old version introduce the cascading sequence and IC model here. I feel that
	it gets into too technical details too quickly, so moved them to the later part,
	at the beginning of the contribution section.}

In contrast to some network processes such as random walks, 
  network influence is defined not solely by the static
  graph structure. 
It is fundamentally defined by the interaction 
  between the {\em dynamic} influence models and 
  the {\em static} network structures.
Even on the same static network, different 
  influence propagation models --- such as the
  popular {\em independent cascade} 
  and {\em linear-threshold} models ---
  induce different underlying relationships 
  among nodes in the network.
The characterization of this interplay
  thus requires us to reformulate various fundamental 
  graph-theoretical concepts such as centrality, closeness, distance,
  neighborhood (sphere-of-influence), and clusterability,
  as well as to identify new concepts fundamental 
  to emerging network phenomena.

In this paper, we will study the following basic question in network
  science with focusing on {\em influence-based network centrality}.

\begin{quote}
{\em Is there a systematic framework to expand graph-theoretical concepts
  in network sciences?
}
\end{quote}

\subsection{Motivations}

Network centrality --- a basic concept in network analysis ---
  measures the importance and the criticality 
  of nodes or edges within a given network.
Naturally, as network applications vary --- 
  being Web search, internet routing, 
  or social interactions --- centrality formulations
  should adapt as well.
Thus, numerous centrality measures have been proposed, 
  based on 
    degree, 
    closeness, 
    betweenness,
    and 
    random-walks
    (e.g., PageRank) (cf.~\cite{NewmanBook})
  to capture the significance of nodes on the Web, 
  in the Internet, and within social networks.
Most of these centrality measures depend only on 
  the static graph structures of the networks.
Thus, these traditional centrality formulations
  could be inadequate for many real-world applications --- 
  including social influence,   viral marketing, and 
  epidemics control ---
  in which static structures are only part of the network
  data that define the dynamic processes.
Our research will focus on the following basic questions:

\begin{quote}
{\em 
How should we summarize influence data to capture the significance
   of nodes in  the dynamic propagation process defined 
  by an influence model?
  How should we extend graph-theoretic centralities to the influence-based centralities?
What does each centrality formulation capture? 
How should we comparatively evaluate different centrality formulations?
}
\end{quote}

At WWW'17, Chen and Teng \cite{CT17} 
   presented an axiomatic framework for 
   characterizing influenced-based network centralities.
Their work is  motivated by studies in multiple disciplines, including 
  social-choice theory \cite{ArrowBook}, 
  cooperative game theory \cite{Shapley53},
  data mining \cite{GhoshInterplay},
   and particularly by
   \cite{IntellectualInfluence} on measures of intellectual influence and
   \cite{PageRankAxioms} on PageRank.
They present axiomatic characterizations
    for two basic centrality measures: 
	(a) {\em Single Node Influence} (SNI) {\em centrality}, 
	which measures each node's significance
  by its influence spread;\footnote{The influence spread ---
               as defined in \cite{kempe03}  --- of 
		a group is the expected number 
		of nodes this group 
		can activate as the initial active set, called {\em seed set}.}
	(b) {\em Shapley Centrality}, which uses the Shapley value 
	of the influence spread function ---
	formulated based on a fundamental 
	cooperative-game-theoretical concept. 
Mathematically, the axioms are structured into two categories.

\begin{itemize}
\item {\bf Principle Axioms}: 
The set of axioms that 
  all desirable influenced-based centrality formulations
  should satisfy.
In \cite{CT17}, two principle axioms, 
  {\sf Anonymity} and  {\sf Bayesian}, are identified.
{\sf Anonymity} is an ubiquitous and exemplary principle axiom,
 which states that
  centrality measures should be preserved under
  isomorphisms among influence instances.
{\sf Bayesian} states that influence-based centrality is a linear measure
  for mixtures of influence instances.
\item 
{\bf Choice axioms:}
A (minimal) set of axioms that together with the principle
  axioms uniquely determine a given centrality formulation.
\end{itemize}

Such characterizations and the taxonomy of axioms
   precisely capture the essence of centrality formulations 
   as well as their fundamental differences.
In particular, the choice axioms 
  succinctly distill
  the comparative differences 
  between different centrality formulations.


\wei{I feel that the motivation and comparison with~\cite{CT17} is too detailed, and
	making the pace slower to get to our contribution. I tried to simplify it a bit,
	but the old version is kept as commented part.}

However, the axiomatic characterization in~\cite{CT17} 
	has two limitations preventing it to be generalized
	to study more influence-based centralities.
First, it makes a significant simplification of the influence process: 
	each influence instance $\I$ only captures the probability distributions of the
	final influenced nodes given any initial seed set,  which we refer as the seeds-targets
	(probabilistic) profile.
Essentially, it compressed out all intermediate steps in a cascading process
	and only takes the initial seed nodes and the final target nodes into account.
This simplification is enough to study centrality measures concerning the
	final influence spread of the diffusion model, but is inadequate for characterizing 
	influenced-based centrality measures that
	can capture the propagation details of network influences, such as neighborhood, closeness,
	sphere-of-influence centralities.
Second, its choice axioms are based on a family of {\em critical set instances}, which
	do not have a graph-theoretical interpretation.
This makes it less powerful in explaining the connection between graph-theoretical
	centralities and influence-based centralities.

%
%
\subsection{Our Contributions}

In this paper, we address both of the above issues in~\cite{CT17}
	and significantly 
  expand the characterization of
  influence-based network centrality.
First, influence instance $\I$ is now defined as the probabilistic
	profile on the more detailed {\em cascading sequences}, as formally defined below.


\begin{restatable}[Cascading Sequence]{mydef}{cascadesequence} \label{def:ICS}
	For a directed network $G = (V,E)$ with $n=|V|$,
	a set sequence $(S_0, S_1, \ldots, S_{n-1})$ 
	is a {\em cascading sequence}
	if it is both
	(1) {\bf Monotonic}: 
	$\emptyset \subset S_0\subset S_1 \subset \cdots \subset  S_{t-1} \subset S_t  = S_{t+1}
	= \cdots S_{n-1} \subseteq V$ for some $t=0, 1, \ldots, n-1$, and 
	(2) {\bf $G$-continuous}:
	for all $t\in [T]$, every node  in $\Delta_t = S_t-S_{t-1}$
	can be reached directly from some nodes in 
	$\Delta_{t-1}$ (where $\Delta_0 = S_0$).
\end{restatable}

In the above definition,  $S_t$ represents the set of network nodes
that become {\em active} by step $t$ during the propagation,
and $\Delta_t = S_t-S_{t-1}$ denotes the
set of nodes newly {\em activated} at step $t$.
Thus, the cascading sequence $(S_0, S_1, \ldots, S_{n-1})$ provides a layered structure
	starting from seed set $S_0$, similar to network broadcasting.


%

However, unlike broadcasting, the layered cascading sequences in 
network influence are formed {\em stochastically.}
In each time step, already activated nodes stochastically activate more nodes in the next
	step, and this stochastic propagation ends when no new nodes are activated in a step.
Therefore, when describing an influence instance, we need to specify the
	probabilistic distribution of the possible cascading sequences.
This is formally defined as influence profile below.
%
%
%

\begin{restatable}[Influence Profile]{mydef}{infprofile} \label{def:infinstance1}
	An {\em influence profile} (also called an {\em influence instance}) is 
	a tuple $\I=(V,E,P_{\I})$, where $G=(V,E)$ is a directed graph, and
	$P_{\I}: (2^V)^n \rightarrow \R$ is a probability profile
	defining a probability distribution of cascading sequences
	for every seed set $S_0 \subseteq V$.
	That is, $P_{\I}(S_0, S_1, \ldots, S_{n-1}) \in [0,1]$
	specifies  the probability 
	that the influence instance generates 
	cascading sequence $(S_0, S_1, \ldots, S_{n-1})$
	when given seed set $S_0$.
\end{restatable}

Then an influence-based centrality measure is defined as a mapping from the above
	influence profiles to real-valued vectors assigning centrality values to each node.
We further transfer the concept of graph distance to cascading distance for cascading sequences.
By using the cascading-sequence based influence profile and cascading
	distance, we are able to extend a 
	large family of graph theoretical centralities based on graph distances from individual nodes,
	such as degree, closeness, harmonic, 
	reachability centralities to influence based centralities.
We refer to them as the stochastic sphere-of-influence centralities, and use
	a generic distance function $f$ to summarize all of them.

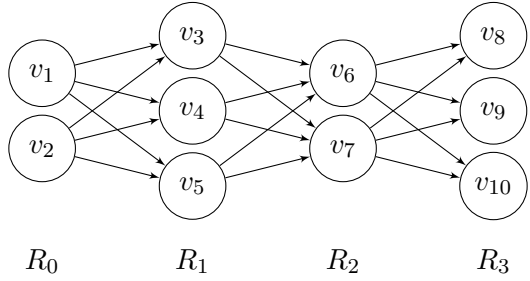
\begin{figure}[t]
	\centering
	\begin{tikzpicture}
	\tikzset{vertex/.style = {shape=circle,draw,minimum size=2.3em}}
	\tikzset{edge/.style = {->,> = latex'}}
	\node[vertex] (v01) at (0, -0.5) {$v_1$};
	\node[vertex] (v02) at (0, -1.5) {$v_2$};
	\node (l0) at (0, -3) {$R_0$};
	\node[vertex] (v11) at (2, 0) {$v_3$};
	\node[vertex] (v12) at (2, -1) {$v_4$};
	\node[vertex] (v13) at (2, -2) {$v_5$};
	\node (l1) at (2, -3) {$R_1$};
	\node[vertex] (v21) at (4, -0.5) {$v_6$};
	\node[vertex] (v22) at (4, -1.5) {$v_7$};
	\node (l2) at (4, -3) {$R_2$};
	\node[vertex] (v31) at (6, 0) {$v_8$};
	\node[vertex] (v32) at (6, -1) {$v_9$};
	\node[vertex] (v33) at (6, -2) {$v_{10}$};
	\node (l3) at (6, -3) {$R_3$};
	\draw[edge] (v01) to (v11);
	\draw[edge] (v01) to (v12);
	\draw[edge] (v01) to (v13);
	\draw[edge] (v02) to (v11);
	\draw[edge] (v02) to (v12);
	\draw[edge] (v02) to (v13);
	\draw[edge] (v11) to (v21);
	\draw[edge] (v11) to (v22);
	\draw[edge] (v12) to (v21);
	\draw[edge] (v12) to (v22);
	\draw[edge] (v13) to (v21);
	\draw[edge] (v13) to (v22);
	\draw[edge] (v21) to (v31);
	\draw[edge] (v21) to (v32);
	\draw[edge] (v21) to (v33);
	\draw[edge] (v22) to (v31);
	\draw[edge] (v22) to (v32);
	\draw[edge] (v22) to (v33);
	\end{tikzpicture}
	\caption{an example of a layered-graph instance with 4 layers $(R_0, R_1, R_2, R_3)$ where $R_0 = \{v_1, v_2\}$, $R_1 = \{v_3, v_4, v_5\}$, $R_2 = \{v_6, v_7\}$, $R_3 = \{v_8, v_9, v_{10}\}$.}
	\label{fig:layered-graph_instance}
\end{figure}

Second, we provide a key technical contribution of the paper, which characterizes
	all influence instances by layered graphs.
A layered graph is a directed graph with multiple layers and all nodes in one layer connect to
	all nodes in the next layer (Fig.~\ref{fig:layered-graph_instance}).
A layered graph instance is simply treating the graph as an influence instance following
	the breadth-first-search (BFS) propagation pattern.
Surprisingly, we show that the set of all layered-graph instances form the basis in the
	vector space of all influence profiles, meaning that every influence profile is
	a linear combination of layered-graph instances.
The result is instrumental to our centrality characterization, and is a powerful result
	by its own right.
	
The above layered-graph characterization allows us to connect influence propagation with
	static graphs.
Therefore, by combining the principle axioms (Axioms {\sf Anonymity} and {\sf Bayesian}),
	we are able to show that 
  our extended stochastic sphere-of-influence centrality 
	with a distance function $f$ is the unique centrality
  satisfying Axioms {\sf Anonymity} and {\sf Bayesian}   
 that	conforms with the corresponding 
  graph-theoretical centrality with the same distance function $f$, and
	the centrality is uniquely determined by their 
  values on layered-graph instances.
This characterization illustrates that 
	(a) our extension of graph-theoretical 
  centralities to influence-based centralities
  is not only reasonable but the only feasible  mathematical 
  choice, and
	(b) layered graphs are the key family of graphs comparing different influence measures,
	since a graph-theoretical centrality measure on layered graphs fully determines
	the conforming influence-based centrality measure.

The above centrality extension focuses on individual nodes.
As group 
  cooperation plays an important role in influence propagation, we further
	extend individual centralities to group centralities, which provide a value for every
	subset of nodes in the network.
Similar to individual centrality, we provide a characterization theorem showing that
	influence-based group centrality is also uniquely characterized by their values
	on layered-graph instances, as long as they satisfy the group version of 
	Axioms {\sf Anonymity} and {\sf Bayesian}.

A group centrality measure has $2^n$ dimensions, so we further use the Shapley value~\cite{Shapley53}
	in cooperative game theory to reduce it to $n$ dimensions, and refer to it as
	the influence-based Shapley centrality.
The Shapley centrality of a node measures its importance when the node collaborate with
	other nodes in groups.
Due to the linearity of the Shapley value, we obtain the same characterization for
	the Shapley centrality: it is the unique one conforming with the graph-theoretical Shapley centrality and satisfying Axioms {\sf Anonymity} and {\sf Bayesian}, and the uniqueness
	is fully determined by its values on layered-graph instances. 

Figure~\ref{fig:map} summarizes our systematic extension of graph-theoretical
	centralities (the lower three boxes) to influence-based centralities (the upper
	three boxes).
Starting from the classical graph-theoretical distance-based individual centrality (e.g.\ harmonic
	centrality), by transferring the concept of graph distance to cascading distance,
	we could lift it to the stochastic sphere-of-influence individual centrality 
	(e.g.\ influence-based harmonic centrality).
From individual centralities (either graph-theoretical or influence-based), we could use
	group distance to extend them to group centralities.
From group centralities, we could apply Shapley value to obtain Shapley centralities.
Therefore, Figure~\ref{fig:map} provides a road map on how to extend many 
	classical graph-theoretical
	centralities to influence-based centralities.

\begin{figure}[t] 
  {\includegraphics[width=1\linewidth]{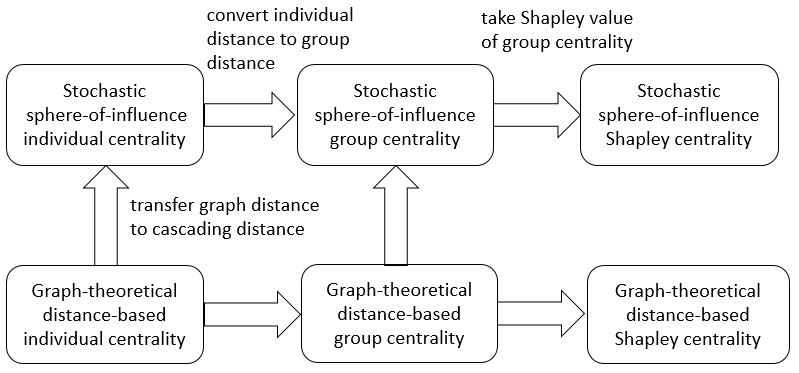}}
  \caption{Road map for the systematic extension of graph-theoretical 
    distance-based centralities to influence-based centralities. \label{fig:map}}
\end{figure}

In addition to studying the systematic framework and characterization of influence-based
	centralities, we also address the algorithmic aspects
   of these centralities.
We extend the approximation algorithm of~\cite{CT17} to cover all 
	stochastic sphere-of-influence centralities introduced in this paper.
The algorithm efficiently provides estimates for the centrality values of all nodes
	with theoretical guarantees.

To summarize, our contributions include:
	(a) a systematic extension of graph-theoretical distance-based centralities
	to corresponding influence-based centralities, and further extending them to group 
	centralities and Shapley centralities;
	(b) a key algebraic characterization of influence profiles by layered-graph instances;
	(c) unique characterization of stochastic sphere-of-influence centralities by
	their layered graph centralities and the principle axioms; and
	(d) an efficient algorithmic framework that approximates influence-based centralities
	for all centralities measures proposed in the paper.


\subsection{Related Work}

Influence propagation is an important topic in network science and
	network mining.
One well studied problem on influence propagation is the influence maximization 
	problem~\cite{domingos01,richardson02,kempe03}, which is to find $k$ seed nodes that
	generate the largest influence in the network.
Influence maximization has been extensively studied for improving its efficiency or
	extending it to various other settings (e.g.~\cite{kempe05,Leskovec07,BharathiKS07,ChenWW10,ChenYZ10}).
However, influence maximization is different from the study of individual centrality
	as pointed out by~\cite{borgatti06}.
Putting into our group centrality context, influence maximization can be viewed as the task
	of finding the group with the largest reachability group centrality.
Our efficient centrality approximation algorithm is inspired from the scalable
	influence maximization based on the reverse-reachable set approach ~\cite{BorgsBrautbarChayesLucier,tang14,tang15}.

Network centrality has been extensively studied and numerous centrality measures
	are proposed~(e.g.~\cite{NewmanBook,Bonacich1972,Katz,Bonacich1987power,PageRank,EverettBorgattiGroup}).
As discussed in the introduction, most centrality measures 
	(including group centrality~\cite{EverettBorgattiGroup} and Shapley
	centrality~\cite{ShapleyValueForCentrality1,GomezShapley})
	are graph theoretical, and
	only focus on the static graph structure.
A recent study by Chen and Teng~\cite{CT17} is the first systematic study combining
	network centrality with influence propagation.
As already detailed in the introduction, our study is motivated by~\cite{CT17}, and we aim
	at overcoming the limitations in~\cite{CT17} and extending the study of influence-based
	centrality to much wider settings. 
In particular, the SNI centrality and Shapley centrality studied in~\cite{CT17} are
	two instances related to reachability in the family of 
	sphere-of-influence centrality measures we cover in this paper.
	
Axiomatic approach has been used to characterize network centrality
	~\cite{Sabidussi66,Nieminen73,BV14,SB16,PageRankAxioms}.
These characterizations are mainly for graph-theoretical centralities.
Again, the study in~\cite{CT17} provides the first axiomatic characterization for
	a couple of influence-based centralities.
The characterization we provide in this paper is novel in the sense that
	it connects general influence profiles with a family of classical layered graphs
	so that our characterization can be based on graph-theoretical centralities.

\section{Algebraic Characterization of Network Influences}

In this section, we present our main technical result --- 
  a surprising discovery during our research --- which provides 
   a graph-theoretical characterization of the space of 
  influence profiles.
Specifically, we identify a simple set of classical graphs,
  and prove that when treated as BFS propagation instances,
  they form a linear basis in the space of
  all stochastic cascading-sequence profiles.  
This graph-theoretical characterization of influence models
  is instrumental to our systematic characterization
  of a family of influence-based network centralities.
Moreover, we believe that this result 
  is also important on its own right, and is potentially useful
	in other settings studying general influence propagation models.


\subsection{Stochastic Diffusion Model}

Stochastic diffusion models describe 
  how information or influence are propagated through a network.
A number of models have been 
  well-studied~(cf.~\cite{kempe03,ChenWeiBook}), and among them 
  independent cascade (IC) and 
  linear threshold (LT) models are most popular ones.
Here, we illustrate the stochastic profile of cascading sequences 
  with  the {\em triggering model}
  of Kempe-Kleinberg-Tardos~\cite{kempe03},
  which includes IC and LT models as special cases.
Triggering model will also be the subject of our algorithmic study.

In a triggering model, the static network structure
  is modeled as a directed graph $G=(V,E)$ with $n=|V|$.
Each node $v\in V$ has a random {\em triggering set} 
  $T(v)$ drawn from distribution $D(v)$
  over subsets of $v$'s in-neighbors $N^-(v)$.
At time $t=0$, triggering sets of 
  all nodes are sampled from 
  their distributions, and nodes in a given
	seed set $S\subseteq V$ are activated.
At any time $t\ge 1$, a node $v$ 
  is activated if some nodes in its triggering set $T(v)$ was 
	activated at time $t-1$.
The propagation continues until no new nodes are activated in a step.

In general, we can describe propagation in networks 
  as a sequence of node activations, as described in the introduction.
For convenience, we restate Definition \ref{def:ICS} here.

\cascadesequence*

Note that $S_t$ in a cascading 
  sequence represents the nodes that are active by time $t$.
The monotonicity requirement in the above 
  definition means that (a) active nodes will not be deactivated.
  and (b) each step should have at least one 
  new active node unless the cascade stops.
This corresponds to the {\em progress} 
  diffusion model in the literature.
Since $S_0$ cannot be empty, and 
  the set must grow by at least one node in each cascading step 
   we have at most $n-1$ cascading steps in the sequence.
The $G$-continuous condition means 
  that the activation of any new node in $\Delta_t$ at time $t$ must be
  partially due to the activation of some nodes 
  in $\Delta_{t-1}$ at the previous time step.
At this level of abstraction, 
  a stochastic diffusion model can be viewed as a probabilistic mechanism
  to generate cascading sequences.
Similar to~\cite{ChenWeiBook}, in this paper we use the distribution of the cascading sequences
	as the general specification of the diffusion model, 
	as defined in Definition~\ref{def:infinstance1} and restated here.
\infprofile*
Note that (1) if $(S_0, S_1, \ldots, S_{n-1})$ 
  is not a valid cascading sequence,
  then $P_{\I}(S_0, S_1, \ldots, S_{n-1})=0$.
(2) For every $S_0$,
 $\sum_{S_1, \ldots, S_{n-1}}P_{\I}(S_0, S_1, \ldots, S_{n-1}) = 1$.
We also use the notation $P_{\I}(S_0)$ to denote the distribution of 
	sequence $(S_1, \ldots, S_{n-1})$ starting from $S_0$.

In influence propagation models, 
  one key metric is the {\em influence spread} of a seed set $S$,
  which is the expected number of nodes activated when $S$ is the seed set.
For an influence profile $\I=(V,E,P_{\I})$, we can define influence spread
	$\sigma_{\I}(S) = \sum_{S_1, \ldots, S_{n-1}} P_{\I}(S, S_1, \ldots, S_{n-1}) \cdot |S_{n-1}|$.

We remark that in \cite{CT17}, a more coarse-grained seed-target profile is used to
study two influence-based centrality measures.
The seed-target profile is only suitable for centrality measures addressing the final influence spread,
	but is not detailed enough to study other centrality extensions including extensions to
	degree, closeness, harmonic centralities, etc.
Therefore, in this paper, we focus on the cascading sequence profile.

%
%
Note that for any directed graph $G=(V,E)$, 
we can equivalently interpret it as a diffusion model, where
the diffusion is carried out by the breadth-first search (BFS).
In particular, for any seed set $S_0$,
we have a deterministic cascading sequence
$(S_0, S_1, \ldots, S_{n-1})$, 
where $S_t$ is all the nodes that can be reached within $t$ steps of BFS.
Thus, the probability profile $P_{\I}$ is 
such that only this BFS sequence has probability $1$, and
all other sequences starting from $S_0$ have probability $0$.
We call this instance the {\em BFS influence instance} 
corresponding to graph $G$, and denote it as $\I^{({\rm BFS})}_G$.

\subsection{A Graph-Theoretical Basis of Influence Profiles}

Mathematically, each influence profile $\I$ of stochastic cascading sequences
  as defined in Definition~\ref{def:infinstance1}, 
  can be represented as a vector of probabilities 
  $\left[ P_{\I}(S_0, S_1, \ldots, S_{n-1})\right]$
  over monotonic cascading sequences.
In other words, the vector contains entries for each monotonic
  cascading sequences.
Because for each set $S_0$,
    all valid cascading sequences add up to $1$, 
    so one entry is redundant. 
We remove the entry $P_{\I}(S_0, S_0, \ldots, S_0)$ from the vector,
  and express it implicitly.
The resulting vector certainly has an exponential 
  number of dimensions, as there are exponential number of monotonic
  set sequences.
We use $M$ to denote its dimension.


The set of ``basis'' graphs are the layered graphs, 
  as depicted in Figure~\ref{fig:layered-graph_instance}.
Formally, 
  for a vertex set $V$, for an integer $t \geq 0$, and  $t+1$ disjoint 
 nonempty subsets $R_0, R_1, \ldots, R_t \subseteq V$,
 a {\em layered graph} $L_V(R_0, \dots, R_t)$ 
  is a directed graph in which every node in $R_{i-1}$
  has a directed edge pointing to every node in $R_{i}$, 
  for $i\in [t]$, and
	the rest nodes in 
  $V \setminus \cup_{i=1}^t R_i$ are isolated nodes 
  with no connections to and from any other nodes.
We say that 
  the BFS influence instance of the layered 
  graph $L_V(R_0, \dots, R_t)$, namely
	$\I^{({\rm BFS})}_{L_V(R_0, \dots, R_t)}$, is a 
  {\em layered-graph instance}, and for convenience we also
	use $\I_V(R_0, \dots, R_t)$ to denote this instance.
When the context is clear, we ignore $V$ in the subscript.
A trivial layered graph instance is when $t=0$, 
  in which case all nodes are isolated and there
	is no edge in the graph.
We call this the {\em null influence instance}, 
  and denote it as $\I^N$ (or $\I_V(V)$ to make
	it consistent with the layered-graph notation).
Technically, in $\I^N$, 
  only $P_{\I^N}(S_0 , S_0, \ldots, S_0) =1$,
   and all other probability values are $0$, which means its
   corresponding vector form is the all-zero vector.


Let $\cL = \{\I_V(R_0, \dots, R_t) \mid t=1, \ldots, n-1, \emptyset \ne R_i \subseteq V,  
	\mbox{all $R_i$'s are disjoint} \}$, i.e. $\cL$ is the set of all nontrivial layered-graph instances
	under node set $V$.

As a fundamental characterization 
  of the mathematical space of influence profiles, we prove the 
  following theorem, which states 
  that all nontrivial layered-graph instances
  form a linear basis in the space of all cascading-sequence based
influence instances:

\begin{restatable}[Graph-Theoretical Basis]{mythm}{layeredgraphbasis} 
\label{thm:basis}
The set of vectors corresponding to the 
  nontrivial layered-graph instances in $\cL$ forms a basis
	in $\R^M$.
\end{restatable}

Although the proof of this theorem is quite technical, its underlying
principle is quite basic.  Here we provide some intuitions.  
Note first that $M = |\cL|$.
Thus, the central argument in the proof is to 
  show that elements in $\cL$ are independent,
  which we will establish using proof-by-contradiction: 
Suppose the profiles
  corresponding to layered graphs are not independent.
 That is, there are not-all-zero coefficients $\lambda_\I$ such that a linear
combination (denoted by $P = \sum_{\I \in \cL} \lambda_\I P_{\I}$) of these profiles is zero. 
We consider a carefully-designed inclusion-exclusion 
  form linear combination of the entries of $P$ and
  show that this combination is exactly some $\lambda_\I \ne 0$, which
means $P \ne \mathbf{0}$.

\wei{We may need some explanation of the proof for this main technical contribution.
	Hanrui, can you draft one for this?}

%
%
%
%
%
%

\section{Influence-based Centrality: Single Node Perspective}
\label{sec:centralitydef}

Recall that a graph-theoretical centrality, 
  such as degree, distance, PageRank, and betweenness,  summarizes
  network data to measure the importance of each node
  in a network structure.
Likewise, the objective of influence-based centrality formulations
  is to summarize the network-influence data in order to 
  measure the importance of every node in influence diffusion processes.
Formally (in the cascading-sequence model): 
\begin{mydef}[Influence-based Centrality Measure] 
\label{def:infcentrality}
An {\em influence-based centrality measure} 
 $\psi$ is a mapping from an influence profile
  $\I = (V,E,P_{\I})$,
  to a real-valued vector $(\psi_v(\I))_{v\in V} \in \R^n$.
\end{mydef}

The objective is to formulate network centrality measures that reflect dynamic
  influence propagation.
Theorem \ref{thm:basis} lays the foundation
  for a systematic framework 
  to generalize graph-theoretical centrality formulation
  to network centrality measures.
To this end, we first examine a unified centrality family that 
  are natural for layered graphs.

\subsection{A Unified Family of Sphere-of-Influence Centralities}


In this subsection, we discuss a family of graph-theoretical
  centrality measures that contains various forms of ``sphere-of-influence''
  and closeness centralities. 
These centrality measures have a common feature: 
  the centrality of node $v$ is fully determined 
  by the distances from $v$ to all nodes.

Consider a directed graph $G=(V,E)$. 
Let $N_G^+(v)$ and $N_G^-(v)$ denote the 
  set of out-neighbors and in-neighbors, 
  respectively, of a node $v$.
Let $d_G(u,v)$ be the graph distance from $u$ to $v$ in $G$.
If $v$ is not reachable from $u$ in $G$ then we set $d(u,v)=\infty$.
Let $d_G(S,v) = \min_{u\in S} d_G(u,v)$ be the distance 
  from a subset $S\subseteq V$ to node $v$.
Let $\Gamma_G(S)$ be the set of nodes reachable from $S\subseteq V$ in $G$.
When the context is clear, we would remove 
  $G$ from the subscripts in the above notations.

Recall that a graph-theoretical centrality measure 
  is a mapping $\mu$ from a graph $G$ to a real-valued
  vector $(\mu_v(G))_{v\in V} \in \R^{n}$, 
  where $\mu_v(G)$ denote the centrality of $v$ in $G$.

For every $S\subseteq V$ and  $v \in V$, let $\vec{d}_G(S)$ be 
   the vector in $\R^n$ consisting of the distance from $S$
	to every node $u$, i.e. $\vec{d}_G(S) = (d(S,u))_{u\in V}$.
Let $\vec{d}_G(v) = \vec{d}_G(\{v\})$.
We use $\R_{\infty}$ to denote $\R \cup \{\infty\}$.
For each $f: \R_{\infty}^n \rightarrow \R$, we can define: 
\begin{mydef}[Distance-based Centrality]\label{def:DC}
A distance-based centrality $\mu^{\Individual}[f]$ 
  with function $f: \R_{\infty}^n \rightarrow \R$
is defined as $\mu^{\Individual}[f]_v(G) = f(\vec{d}_G(v))$.
\end{mydef}

Definition \ref{def:DC} is a general formulation.
It includes several classical graph-theoretical centrality
  formulations as special cases:
(a) The {\em degree centrality} (or immediate sphere of influence), 
  $\mu^{\IndDegree}$, is 
  defined as the out-degree of a node $v$ in graph $G$,
	that is, $\mu_v^{\IndDegree}(G) = |N_G^+(v)|$.
It is defined by
  $f^{\Degree}(\vec{d}) = |\{u\in V \mid d_u = 1 \} |$.
(b) The {\em closeness centrality}, $\mu^{\Individual\text{-}\Closeness}$, 
  is defined as the reciprocal of the average distance
	to other nodes, 
   $\mu_v^{\Individual\text{-}\Closeness}(G) = 
  \frac{1}{\sum_{u\ne v} d_G(v,u)}$.
It is defined by  $f^{\Closeness}(\vec{d}) = 
\frac{1}{\sum_{u\in V} d_u}$.
If $G$ is not strongly connected, then 
   $\mu_v^{\Individual\text{-}\Closeness}(G) = 0$ 
  for any $v$ that cannot reach all
	other nodes, and thus closeness centrality is not expressive enough for such graphs.
(c) {\em harmonic centrality}, $\mu^{\Individual\text{-}\Harmonic}$, is defined:
$\mu_v^{\Individual\text{-}\Harmonic}(G) = \sum_{u\ne v} \frac{1}{d_G(v,u)}$.
It is defined by 
  $f^{\Harmonic}(\vec{d}) = \sum_{u\in V, d_u > 0} \frac{1}{d_u}$.
Note that harmonic centrality is closely related with closeness centrality,
 and is applicable to network with disjointed components.
(d) The {\em reachability centrality} measure 
   $\mu^{\Individual\text{-}\Reachability}$:
  $\mu_v^{\Individual\text{-}\Reachability}(G) = |\Gamma_{G}(\{v\})|$, 
  which means the reachability centrality of $v$
	is the number of nodes $v$ could reach in $G$.
It is defined by
  $f^{\Reachability}(\vec{d}) = |\{u\in V \mid d_u < \infty  \} |$.
(e) The {\em sphere-of-influence centrality} measure 
  $\mu^{\Individual\text{-}\SoI(\delta)}$:
For a threshold parameter $\delta$,
   $\mu_v^{\Individual\text{-}\SoI(\delta)} = |\{u: d_u \leq \delta\}|$.
It is defined by
  $f^{\SoI(\delta)}(\vec{d}) = |\{u\in V \mid d_u \le \delta \} |$.
Clearly, as $\delta$ varies from $1$ to $n-1$ (or $\infty$), the 
  sphere-of-influence centrality interpolates
  the degree centrality and the reachability centrality. 

\subsection{Stochastic Sphere-of-Influence: Lifting from Graph to Influence Models}




Thus, Definition \ref{def:DC} represents a unified
   family of sphere-of-influence centralities for graphs.
The function $f$ --- which is usually a non-increasing function of distance
  profiles --- captures the {\em scale} of the impact, based on the 
  distance of nodes from the {\em source}.
By unifying these centralities under one general centrality class, 
  we are able to systematically derive and  study their generalization 
  in the network-influence models.
The key step is to transfer the graph distance in directed graph 
  to {\em cascading distance}
	in cascading sequences.
For any cascading sequence 
  $(S_0, S_1, \ldots, S_{n-1})$ starting from seed set $S_0$, 
  let $d_u(S_0, S_1, \ldots, S_{n-1})= t$ 
  if $u\in \Delta_t = S_t \setminus S_{t-1}$ ($\Delta_0=S_0$), 
  and $d_u(S_0, S_1, \ldots, S_{n-1})= \infty$ if $u \not\in S_{n-1}$.
We call $d_u(S_0, S_1, \ldots, S_{n-1})$ the {\em cascading distance} 
  from seed set $S_0$ to node $u$,
  since it represents the number of steps 
  needed for $S_0$ to activate $u$ in the cascading sequence.
Then, we define the cascading 
  distance vector $\vec{d}(S_0, S_1, \ldots, S_{n-1})$ 
  as $(d_u(S_0, S_1, \ldots, S_{n-1}))_{u\in V} \in \R^n$.

In particular, when we consider 
  a cascading sequence $(\{v\}, S_1, \ldots, S_{n-1})$
  starting from a single node $v$, 
  set $\Delta_1 = S_1 \setminus \{v\}$ can be 
  viewed as the out-neighbor of $v$,
  set $S_{n-1}$ can be viewed as the all nodes reachable from $v$,
  and for every node $u \in \Delta_t = S_t \setminus S_{t-1}$, 
  the distance from $v$ to $u$	is $t$.

\begin{mydef}[Individual Stochastic Sphere-of-Influence Centrality] 
  \label{def:distfunc}
For each function $f: \R_\infty^n \rightarrow \R$, 
   the influence-based individual stochastic sphere-of-influence centrality $\psi^{\Individual}[f]$
   is defined as:
\begin{align*}
&\psi^{\Individual}[f]_v(\I) = \E_{(S_1, \ldots, S_{n-1})\sim P_{\I}(\{v\})}[f(\vec{d}(\{v\}, S_1, \ldots, S_{n-1}))].
\end{align*}
\end{mydef}


Definition \ref{def:distfunc} systematic extended the family
  of graph-theoretical centralities of Definition \ref{def:DC} to
  influence models.
Natural influence-based centralities, e.g., the 
  single-node influence (SNI) centrality defined in~\cite{CT17}
  (using  each node $v$'s influence spread $\sigma_{\I}(\{v\})$
  as the measure of its influence-based centrality),
  can be expressed by this extension:
\begin{restatable}{myprop}{snireachability}
$\forall$ influence profile $\I = (V,E,P_{\I})$:
\[\SNI(\I) = \psi^{\Individual}[f^{\Reachability}](\I).\]
\end{restatable}
The influence-based centrality formulations 
  of Definition \ref{def:distfunc} 
  enjoy the following graph-theoretical conformity property.
\begin{mydef}[Graph-Theoretical Conformity]
An influence-based centrality measure $\psi$
  conforms with a graph-theoretical centrality measure $\mu$
   if for every directed graph $G$, $\psi(\I^{({\rm BFS})}_G) = \mu(G)$.
\end{mydef}

\begin{restatable}{myprop}{conformdistfunc}
\label{prop:conformdistfunc}
For any function $f: \R_\infty^n \rightarrow \R$, 
 $\psi^{\Individual}[f]$  conforms with  $\mu^{\Individual}[f]$.
\end{restatable}

\subsection{Characterization of Influence-Based Centrality Measures}

%

Given the multiplicity of the (potential) centrality formulations,
  ``how should we characterize each formulation?'' and ``how should
    we compare different formulations?'' are fundamental questions
    in network analysis.
Inspired by the pioneering work of 
  Arrow \cite{ArrowBook} on social choice, 
  Shapley \cite{Shapley53} on cooperation games and coalition,
   Palacios-Huerta \& Volij  \cite{IntellectualInfluence} 
  on measures of intellectual influence, and
  Altman \& Tennenholtz  \cite{PageRankAxioms} on PageRank,
  Chen and Teng \cite{CT17} 
  proposed an axiomatic framework for characterizing and analyzing
  influence-based network centrality.
They identify two principle axioms that 
   all desirable influenced-based centrality formulations
   should satisfy.

\vspace*{0.05in}
\noindent{\sc Principle Axioms for Influenced-Based Centrality}
\vspace*{0.05in}

The first axiom --- ubiquitous axiom for centrality characterization, 
  e.g.~\cite{Sabidussi66}
   ---  states that
  labels on the nodes should have no effect on centrality measures.
\begin{axiom}[Anonymity] \label{axiom:anonymity}
For any influence instance $\I = (V,E,P_{\I})$, and permutation 
$\pi$ of $V$, 
\begin{eqnarray}\label{eqn:anonymity}
\psi_{v}(\I) = \psi_{\pi(v)}({\pi(\I)}), \quad \forall v \in V.
\end{eqnarray}
\end{axiom}
\vspace{-1mm}
In Axiom \ref{axiom:anonymity},
  $\pi(\I) = (\pi(V),\pi(E), \pi(P_{\I}))$
denotes the isomorphic instance:
(1) $\forall u,v\in V$,  $(\pi(u),\pi(v))\in \pi(E)$ if and only if  $(u,v)\in E$, and
(2) for any cascading sequence $(S_0, S_1, \ldots, S_{n-1})$, 
$P_{\pi(\I)}(\pi(S_0), \pi(S_1), \ldots, \pi(S_{n-1})) 
 =  P_{\I}(S_0, S_1, \ldots, S_{n-1})$.


The second axiom concerns {\em Bayesian social influence} 
   \cite{CT17}
  through a given network:
For any three influence profiles 
 $\I$, $\I_1$, $\I_2$ over  the same vertex set $V$,
  we say $\I$ is a {\em Bayesian} of $\I_1$ and $\I_2$
  if there exists $\alpha \in [0, 1]$ such that 
  $ P_{\I} = \alpha P_{\I_1} + (1 - \alpha) P_{\I_2}$.
In other words, $\I$ can be interpreted 
  as a stochastic diffusion model where we first 
  make a random selection --- with probability $\alpha$ 
  of model $\I_1$ and with probability $(1-\alpha)$ of model $\I_2$ --- 
  and then carry out the diffusion process according to the selected model.
We also say that $\I$ is a {\em convex combination} of $\I_1$ and $\I_2$.
The  axiom reflects the linearity-of-expectation principle.
If an influence
  instance is a convex combination of two other influence
  instances, the centrality value of a vertex 
  is the same convex combination of the corresponding centrality
  values in the two other instances.
\begin{axiom}[Bayesian]
For any $\alpha\in [0,1]$, for
   any influence profiles $\I$, $\I_1$ and $\I_2$ over common vertex set $V$
   such that
  $ P_{\I} = \alpha P_{\I_1} + (1 - \alpha) P_{\I_2}$,
\begin{eqnarray}\label{eqn:Bayesian}
   \psi_v(\I) = \alpha \psi_v(\I_1) + (1 - \alpha) \psi_v(\I_2), \quad 
  \forall v \in V.
\end{eqnarray}
\end{axiom}






\vspace*{0.05in}
\noindent{\sc Characterization of Influence-Based Centrality}
\vspace*{0.05in}

We now use Theorem \ref{thm:basis} and 
  Axioms {\sf Anonymity} and {\sf Bayesian}
  to establish a complete characterization  of the family of 
  stochastic sphere-of-influence centralities formulated in 
  Definition \ref{def:distfunc}.
  
\begin{restatable}{myprop}{distfuncanonymous}
\label{prop:distfuncanonymous}
If a function $f: \R_\infty^n \rightarrow \R$ is {\em anonymous}  --- i.e., 
  $f(\vec{d})$ is {\em permutation-invariant} ---
  then $\psi^{\Individual}[f]$ (as defined in Definition \ref{def:distfunc})
  satisfies  Axiom {\sf Anonymity}
  and  $\mu^{\Individual}[f]$ (as defined in Definition \ref{def:DC})
  satisfies the graph-theoretical counterpart of Axiom {\sf Anonymity}.
\end{restatable}


\begin{restatable}{myprop}{distfuncbayesian}
\label{prop:distfuncbayesian}
For any function $f: \R_\infty^n \rightarrow \R$, 
  $\psi^{\Individual}[f]$  satisfies  Axiom {\sf Bayesian}.
\end{restatable}

Theorem~\ref{thm:basis}  shows that all influence profiles
   can be represented as a linear combination
   of nontrivial layered-graph instances.
The result enables us to study and 
   compare centrality measures by looking at their instantiation
   in the simple layered graph instances.
The Bayesian property together with the linear basis of nontrivial
  layered-graph instances leads to the following characterization theorem.
\begin{restatable}{mythm}{uniquedetermine}
\label{thm:unique_determination}
A Bayesian influence-based centrality measure is 
  uniquely determined by its values on 
  layered-graph instances (including the null instance).
\end{restatable}
Since layered-graph instances are all BFS 
  instances derived from a special
  class of directed graphs, 
  a direct implication of Theorem~\ref{thm:unique_determination}
  is that any Bayesian centrality conforming 
  with a classical graph-theoretical centrality is unique.
Therefore, we have:
\begin{restatable}[Characterization of Individual Centrality]{mythm}{IndividualCharacterization}
\label{thm:IndividualCharacterization}
For any anonymous function $f: \R_\infty^n \rightarrow \R$, 
  $\psi^{\Individual}[f]$  (defined in Definition~\ref{def:distfunc})
   is the unique influence-based centrality
   that conforms with $\mu^{\Individual}[f]$  
  (defined in Definition~\ref{def:DC})
   that satisfies both Axiom {\sf Anonymity}   and Axiom {\sf Bayesian}.
\end{restatable}

%

The above uniqueness result show that 
  our generalized definitions of influence-based
  centralities Definition~\ref{def:distfunc} are not only reasonable
  but the only feasible choice (to the extent of Bayesian centralities).

\section{Influence-based Centrality: Group Perspective and Shapley Centrality}


As highlighted in
  Domingos-Richardson \cite{RichardsonDomingos,DomingosRichardson}
  and Kempe-Kleinberg-Tardos \cite{kempe03},
  social influence propagation and viral marketing
  are largely group-based phenomena.
Besides characterizing individuals' influential centralities,
  perhaps the more important task is to characterize the influential
  centrality of  groups, and individuals' roles in group cooperation.
This is the group centrality and Shapley centrality 
  introduced in this section.
When distinction is necessary, 
  we refer to the centrality defined 
  in Section~\ref{sec:centralitydef} as {\em individual centrality}.

\subsection{Group Centrality}

Group centrality measures the importance of each group in a network.
Formally, 
\begin{mydef}[Influence-based Group Centrality] 
  \label{def:groupcentrality}
An {\em influence-based group centrality measure} 
 $\psi^{\Group}$ is a mapping from an influence profile
   $\I = (V,E,P_{\I})$ to a 
  real-valued vector  $(\psi^{\Group}_S(\I))_{S\in 2^V} \in \R^{2^n}$.
\end{mydef}

For any function $f: \R_{\infty}^n \rightarrow \R$, 
  both Definition \ref{def:DC} and Definition \ref{def:distfunc}
  have a natural extension: For  $S \subseteq V$, 
\begin{eqnarray*}
\mu^{\Group}[f]_S(G)  =  f(\vec{d}_G(S)),
\end{eqnarray*}
\begin{eqnarray*}
\psi^{\Group}[f]_S(\I)  =  \E_{(S_1, \ldots, S_{n-1})\sim P_{\I}(\{v\})}[f(\vec{d}(S, S_1, \ldots, S_{n-1}))].
\end{eqnarray*}

Axioms {\sf Anonymity} and {\sf Bayesian} extend naturally 
  as well as 
  the following characterization based on Theorem~\ref{thm:basis}.


\begin{restatable}[Characterization of Group Centrality]{mythm}{GroupCharacterization}
\label{thm:GroupCharacterization}
For any anonymous function $f: \R_\infty^n \rightarrow \R$, 
  $\psi^{\Group}[f]$ 
   is the unique influence-based group centrality
   that conforms with $\mu^{\Group}[f]$  and 
   satisfies both Axiom {\sf Anonymity}  and Axiom {\sf Bayesian}.
\end{restatable}


	
Therefore, we can again reduce the analysis
   of an influence-based group centrality 
   to the analysis of the measure on the particular layered-graph
	instances. 

\subsection{Cooperative Games and Shapley Value}

A {\em cooperative game} \cite{Shapley53}
   is defined by tuple $(V, \tau)$, where $V$ is a set of $n$ players, and 
   $\tau:2^V \rightarrow \R$ is called 
   {\em characteristic function} specifying the
	cooperative utility of any subset of players.
In cooperative game theory, 
  a ranking function $\phi$ is a mapping from a characteristic function $\tau$
  to a vector $(\phi_v(\tau))_{v\in V} \in \R^n$, 
  indicating the importance of each individual in the cooperation.
One famous ranking function is 
  the Shapley value $\phi^{\Shapley}$~\cite{Shapley53}, as defined below.
Let $\Pi$ be the set of all permutations of $V$, and $\pi \sim \Pi$ denote a random permutation
	$\pi$ drawn uniformly from set $\Pi$.
For any $v\in V$ and $\pi \in \Pi$, 
let $S_{\pi, v}$ denote the set of nodes in $V$
	preceding $v$ in permutation $\pi$.
Then, $\forall v\in V$:
\vspace{-1mm}
\begin{align*}
\phi^{\Shapley}_v(\tau) = & \frac{1}{n!}\sum_{\pi \in \Pi} \left( \tau(S_{\pi,v} \cup \{v\}) - \tau(S_{\pi,v})\right) \\
= & \sum_{S\subseteq V \setminus \{v\}} \frac{|S|!(n-|S|-1)!}{n!}
\left( \tau(S \cup \{v\}) - \tau(S)\right) \\
= & \E_{\pi\sim \Pi}[\tau(S_{\pi,v} \cup \{v\}) - \tau(S_{\pi,v})].
\end{align*}
The Shapley value of a player $v$ measures
   the expected marginal contribution of $v$
  on the set of players ordered before $v$ in a random order.
\newcommand{\remove}[1]{}
Shapley 
 \cite{Shapley53} proved a \remove{the following}remarkable representation theorem:
  The Shapley value is the unique ranking function
	that satisfies all the following four conditions:
\remove{\begin{itemize}
\item }(1) {\sf Efficiency}:  
  $\sum_{v\in V} \phi_{v}(\tau) = \tau(V)$.
\remove{\item }(2) {\sf Symmetry}:
For any $u,v\in V$, 
  if  $\tau(S\cup \{u\}) = \tau(S\cup \{v\})$, $\forall 
  S\subseteq V\setminus \{u,v\}$,
 then $\phi_u(\tau) = \phi_v(\tau)$.
\remove{\item }(3) {\sf Linearity}: 
For any two characteristic functions $\tau$ and $\omega$,  for any
 $\alpha, \beta > 0$, 
$\phi(\alpha\tau + \beta\omega) = \alpha\phi(\tau) + \beta\phi(\omega)$.
\remove{\item }(4) {\sf Null Player}:
For any $v\in V$, 
	if $\tau(S\cup \{v\}) - \tau(S) = 0$,  $\forall 
   S\subseteq V \setminus \{v\}$, 
   then~$\phi_v(\tau)=0$.
\remove{\end{itemize}}{\sf Efficiency} states that the total utility is fully distributed.  
{\sf Symmetry} states that two players' ranking values should be the same
 if they have the identical marginal utility profile. 
{\sf Linearity} states that the ranking values of the weighted sum of 
  two coalition games is the same as the weighted
  sum of their ranking values.
{\sf Null Player}
  states that a player's ranking value should be zero if 
  the player has zero marginal utility to every subset.

\subsection{Shapley Centrality}


Shapley's celebrated concept --- 
  as highlighted in \cite{CT17} ---
  offers a formulation for assessing individuals' 
  performance in group influence settings.
It can be used to systematically compress 
  exponential-dimensional group centrality measures
  into $n$-dimensional individual centrality measures.

\begin{mydef}[Influence-based Shapley Centrality]
An {\em influence-based} Shapley centrality 
  $\psi^{\Shapley}$ is an individual centrality 
  measure corresponding to a group centrality $\psi^{\Group}$:
\begin{eqnarray*}
\psi^{\Shapley}_v(\I) & =  & \phi^{\Shapley}_v(\psi^{\Group}(\I))\\
 & = &  \E_{\pi\sim \Pi}[\psi^{\Group}_{S_{\pi,v} \cup \{v\}}(\I) - \psi^{\Group}_{S_{\pi,v}}(\I)].
\end{eqnarray*}
We also denote it as $\psi^{\Shapley} = \phi^{\Shapley} \circ \psi^{\Group}$.
\end{mydef}

In \cite{CT17}, Chen and Teng analyze the 
  Shapley value of the influence-spread function, which is 
  a special case of the following ``Shapley extension'' of 
  Definition \ref{def:distfunc}.
\begin{mydef}[Shapley Centrality of Stochastic Sphere-of-Influence] 
\label{def:distfuncShapley}
For each $f: \R_\infty^n \rightarrow \R$, 
  the {\em Shapley centrality of Stochastic Sphere-of-Influence} 
  $\psi^{\Shapley}[f]$ is defined as: 
\begin{align*}
&\psi^{\Shapley}[f]_v(\I) = \phi^{\Shapley}_v(\psi^{\Group}[f](\I)).
\end{align*}
\end{mydef}

Shapley centrality $\mu^{\Shapley}$ 
  can also be defined similarly based on graph-theoretical group centrality
   (see, for example, 
        \cite{ShapleyValueForCentrality1}).
We will refer to the extension of Definition \ref{def:DC} 
  as  $\mu^{\Shapley}[f]$.
Using Theorem~\ref{thm:basis}, we can establish the 
  following characterization.
\begin{restatable}[Characterization of Shapley Centrality]{mythm}{ShapleyCharacterization}
\label{thm:ShapleyCentrality}
For any anonymous function $f: \R_\infty^n \rightarrow \R$, 
  $\psi^{\Shapley}[f]$  is the unique influence-based centrality
   that conforms with $\mu^{\Shapley}[f]$
  and satisfies both Axiom {\sf Anonymity}  and Axiom {\sf Bayesian}.
\end{restatable}
Theorem \ref{thm:ShapleyCentrality} systematically extends
  the work of \cite{CT17}
  to all sphere-of-influence formulations.
The SNI and Shapley centrality analyzed in \cite{CT17}
  are  $\psi^{\Individual}[f^{\Reachability}]$ and 
   $\psi^{\Shapley}[f^{\Reachability}]$, respectively.
In our process of generalizing the work of \cite{CT17},
  we also {\em resolve an open question} central to
  the axiomatic characterization of \cite{CT17}, 
  which is based on a family of {\em critical set instances}
  that do not correspond to a
  graph-theoretical interpretation.
In fact, the influence-spread functions of these
   ``axiomatic'' critical set instances used in \cite{CT17}
   are not submodular.
In contrast, influence-spread functions of the
  popular  independent cascade (IC) and 
  linear threshold (LT) models, as well as, 
  the trigger models of Kempe-Kleinberg-Tardos, 
  are submodular.
The submodularity of these influence-spread functions
  plays an instrumental role in 
  influence maximization algorithms \cite{kempe03,ChenWeiBook}.
Thus, it is a fundamental and mathematical question whether
  influence profiles can be characterized by 
  ``simpler'' influence instances. 
 Our layered-graph characterization (Theorem~\ref{thm:basis}) resolves this
 	open question by connecting all influence profiles with
 	simple BFS cascading sequence in the layered graphs, which
 	is a special case of the IC model and possess the submodularity property.
In summary, our layered-graph characterization is instrumental to
	the series of characterizations we could provide in this paper for
	influence-based individual, group, and Shapley centralities
	(Theorem~\ref{thm:IndividualCharacterization}, \ref{thm:GroupCharacterization},
	 \ref{thm:ShapleyCentrality}).

\section{Efficient Algorithm for Approximating Influence-based Centrality}

Besides studying the characterization of influence-based centralities, we also want to
	compute these centrality measures efficiently.
Accurate computation is in general infeasible (e.g.\ it is \#P-hard to compute
	influence-based reachability centrality $\psi^{\Individual}[f^{\Reachability}]$ in the triggering model~\cite{wang2012scalable,ChenWW10b}).
Thus, we are looking into approximating centrality values.
Instead of designing one algorithm for each centrality, we borrow the algorithmic
	framework from \cite{CT17} and show how to adapt the framework to approximate
	different centralities.
Same as in \cite{CT17}, the algorithmic framework applies to the triggering model
	of influence propagation.
For efficient computation, we further assume that
	the distance function $f$ is {\em additive}, i.e. 
	 $f(\vec{d}) = \sum_{u\in V} g(d_u)$ for some scalar function $g:\R_{\infty} \rightarrow \R$
	 satisfying $g(\infty) = 0$.
The degree, harmonic, and reachability centralities all satisfy this condition.
In particular, we have $f^{\Degree}(\vec{d}) = \sum_{u\in V} g^{\Degree}(d_u)$,
	with $g^{\Degree}(d_u) = 1$ if $d_u=1$ and $g^{\Degree}(d_u) = 0$ otherwise;
	$f^{\Harmonic}(\vec{d}) = \sum_{u\in V} g^{\Harmonic}(d_u)$, with
	$g^{\Harmonic}(d_u) = 1/d_u$ if $d_u>0$ and $g^{\Harmonic}(d_u) = 0$ otherwise; and
	$f^{\Reachability}(\vec{d}) = \sum_{u\in V} g^{\Reachability}(d_u)$, where
	$g^{\Reachability}(d_u) = 1$ if $d_u < \infty$ and $g^{\Reachability}(d_u) = 0$ otherwise.

\begin{algorithm}[h!]
	\centering
	\caption{\ICERR: Efficient estimation of sphere-of-influence centralities via RR-sets, for the triggering model and additive distance function
	$f(\vec{d}) = \sum_{u\in V} g(d_u)$.} \label{alg:rrcentrality}
	\begin{algorithmic}[1]
		\INPUT{Network: $G=(V,E)$; Parameters:  random triggering set distribution $\{\bT(v)\}_{v\in V}$, 
			$\varepsilon > 0$, $\ell>0$, $k\in [n]$, node-wise distance function $g$
		}
		\OUTPUT{$\epsi_v$,  $\forall v\in V$: estimated centrality value
		}
		\STATE \{Phase 1. Estimate the number of RR sets needed \} \label{line:phase1b}
		\STATE $\LB = 1$; $\varepsilon' = \sqrt{2} \cdot \varepsilon$; $\theta_0 = 0$
		\STATE $\est_v = 0$ for every $v\in V$
		\FOR {$i=1$ to $ \lfloor \log_2 n \rfloor - 1$} \label{line:phase1forb}
		\STATE $x = n / 2^i$
		\STATE $\theta_i = \left\lceil \frac{ n \cdot 
			((\ell + 1)\ln n + \ln \log_2 n + \ln 2) \cdot (2+\frac{2}{3}\varepsilon')}
		{\varepsilon'^2 \cdot x} \right \rceil $ \label{line:setthetai}
		\FOR {$j = 1$ to $\theta_i - \theta_{i-1}$}
		\STATE generate a random RR set $R_v$ rooted at $v$, and for each $u\in R_v$, record the
			distance $d_{R_v}(u,v)$ from $u$ to $v$ in this reverse simulation.
			\label{line:generateRR1}
		\IF {estimating individual centrality}
			\STATE for every $u\in R_v$, $\est_{u} = \est_{u} + g(d_{R_v}(u,v))$ \label{line:estnew1}
		\ELSE 
			\STATE \{estimating Shapley centrality\}
			\STATE for every $u\in R_v$, $\est_{u} = \est_{u} + 
				\phi^{\Shapley}_u(g(d_{R_v}(\cdot,v)))$
			\label{line:estnewShapley1}
		\ENDIF \label{line:endestimate1}
		\ENDFOR
		\STATE $\est^{(k)} = $ the $k$-th largest value in $\{\est_v\}_{v\in V}$ \label{line:kmaxest}
		\IF{ $ n\cdot \est^{(k)} / \theta_i \ge (1+\varepsilon') \cdot x$} \label{line:cond1}
		\STATE $ \LB = n\cdot \est^{(k)} / (\theta_i \cdot (1+\varepsilon'))$ \label{line:setLB}
		\STATE {\bf break}
		\ENDIF
		\ENDFOR \label{line:phase1fore}
		\STATE $\theta = \left\lceil \frac{n ((\ell+1)\ln n + \ln 4)(2+ \frac{2}{3} \varepsilon) }{\varepsilon^2 \cdot \LB} \right\rceil$ \label{line:thetanew}
		\label{line:phase1e}
		\STATE \{Phase 2. Estimate the centrality value\} \label{line:phase2b}
		\STATE $\est_v = 0$ for every $v\in V$ \label{line:resetest}
		\FOR {$j = 1$ to $\theta$ }
		\STATE generate a random RR set $R_v$ rooted at $v$, and for each $u\in R_v$, record the
			distance $d_{R_v}(u,v)$ from $u$ to $v$ in this reverse simulation.
			\label{line:generateRR2}
		\IF {estimating individual centrality}
			\STATE for every $u\in R_v$, $\est_{u} = \est_{u} + g(d_{R_v}(u,v))$ \label{line:estnew2}
		\ELSE 
		\STATE for every $u\in R_v$, $\est_{u} = \est_{u} + 
				\phi^{\Shapley}_u(g(d_{R_v}(\cdot,v)))$
				\label{line:estnewShapley2}
		\ENDIF \label{line:endestimate2}
		\ENDFOR
		\STATE for every $v\in V$, $\epsi_v = n \cdot \est_v / \theta$ 
		\label{line:adjustnew2}
		\STATE return $\epsi_v$, $v\in V$ \label{line:phase2e}
	\end{algorithmic}
\end{algorithm}

The algorithmic framework for estimating individual and Shapley forms of sphere-of-influence 
centrality is given in Algorithm~\ref{alg:rrcentrality}, and is denoted
	{\ICERR} (for Influence-based Centrality Estimate via RR set).
The algorithm uses the approach of reverse-reachable sets (RR sets) \cite{BorgsBrautbarChayesLucier,tang14,tang15}.
An RR set $R_v$ is generated by randomly selecting a node $v$ (called the {\em root} of $R_v$) with
	equal probability, and then reverse simulating the influence propagation starting from $v$.
In the triggering model, it is simply sampling a random triggering set $T(v)$ for $v$, 
	putting all nodes in $T(v)$ into $R_v$, and then recursively sampling triggering sets for all
	nodes in $T(v)$, until no new nodes are generated.

The algorithm has two phases.
In the first phase (lines~\ref{line:phase1b}--\ref{line:phase1e}), the number $\theta$ 
	of RR sets needed for the estimation is computed.
The mechanism for obtaining $\theta$ follows the IMM algorithm in~\cite{tang15} and is
	also the same as in~\cite{CT17}.
In the second phase (lines~\ref{line:phase2b}--\ref{line:phase2e}), $\theta$ RR
	sets are generated, and for each RR set $R$, the centrality estimate of $u\in R$,
	$\est_u$, is updated properly depending on the centrality type.

Comparing to the algorithm in \cite{CT17}, our change is in
	lines~\ref{line:generateRR1}--\ref{line:endestimate1} and
	lines~\ref{line:generateRR2}--\ref{line:endestimate2}.
First, when generating an RR set $R_v$, we not only stores the nodes, but for each $u\in R_v$,
	we also store the distance from $u$ to root $v$ in the reverse simulation paths
	$d_{R_v}(u,v)$.
Technically, $d_{R_v}(u,v)$ is the graph distance from $u$ to $v$ in the subgraph $G_{R_v}$, 
	where $G_{R_v} = (V, E_{R_v})$ with $E_{R_v} = \{(w,u) \mid u\in R_v, w \in T(u) \}$ is the subgraph
	generated by the triggering sets sampled during the reverse simulation.
Note that with this definition, for $u\not\in R_v$, we have $d_{R_v}(u,v) = \infty$.
Next, if we are estimating individual centrality, we simply update the estimate $\est_u$
	by adding $g(d_{R_v}(u,v))$.
If we are estimating Shapley centrality, we need to update $\est_u$ by adding
	$\phi^{\Shapley}_u(g(d_{R_v}(\cdot,v)))$, the Shapley value of $u$ on the
	set function $g(d_{R_v}(\cdot,v)): S\in 2^V \mapsto g(d_{R_v}(S,v)) \in \R$.
We will show below that 
	the computation of $\phi^{\Shapley}_u(g(d_{R_v}(\cdot,v)))$ for all $u\in R_v$ 
	together is linear to $|R_v|$, so it is in the same order of generating $R_v$ and
	does not incur significant extra cost.
Note that the algorithm in~\cite{CT17} corresponds to our algorithm with $g=g^{\Reachability}$.
The correctness of the algorithm replies on the following crucial lemma.

\begin{restatable}{mylem}{RRsetCrucial} \label{lem:RRsetCrucial}
Let $R_v$ be a random RR set with root $v$ generated in a triggering model instance $\I$.
Then, $\forall u\in V$, $u$'s stochastic sphere-of-influence individual centrality with function
	$f(\vec{d}) = \sum_{u\in V} g(d_u)$ is
	$\psi[f]_u(\I) = n\cdot \E[g(d_{R_v}(u,v))]$, where the expectation is taking over 
	the distribution of RR set $R_v$.
Similarly, $u$'s influence-based Shapley centrality with $f$
	is $\psi^{\Shapley}[f]_u(\I) = n\cdot \E[\phi^{\Shapley}_u(g(d_{R_v}(\cdot,v)))]$.
\end{restatable}
From Lemma~\ref{lem:RRsetCrucial}, we can understand that
	lines~\ref{line:estnew2} and~\ref{line:estnewShapley2} are simply accumulating 
	empirical values of $g(d_{R_v}(u,v))$ and $\phi^{\Shapley}_u(g(d_{R_v}(\cdot,v)))$
	for individual centrality and Shapley centrality, respectively, and 
	line~\ref{line:adjustnew2} averages this cumulative value and then multiply it by $n$
	to obtain the final centrality estimate.
With the above lemma, the correctness of the algorithm {\ICERR} is shown by the following
	theorem.

\begin{restatable}{mythm}{thmAlgo} 
	\label{thm:cenest}
	Let $(\psi_v)_{v\in V}$ be the true centrality value for an
		influence-based individual or Shapley centrality with additive function $f$,
		and let $\psi^{(k)}$ be the $k$-th largest value in $(\psi_v)_{v\in V}$.
	For any $\epsilon > 0$, $\ell > 0$, and $k\in [n]$, 
	Algorithm {\ICERR} 
	returns the estimated centrality $(\epsi_v)_{v\in V}$ that satisfies
	(a) unbiasedness: $\E[\epsi_v] = \psi_v, \forall v\in V$; and
	(b) robustness: under the condition that $\psi^{(k)}\ge 1$, with probability at least $1-\frac{1}{n^\ell}$:
	\vspace{-1mm}
	\begin{equation} \label{eq:relativeerror}
	\left\{ 
	\begin{array}{lr}
	|\epsi_v - \psi_v| \le \varepsilon \psi_v & \forall v\in V \mbox{ with } \psi_v > \psi^{(k)},\\
	|\epsi_v - \psi_v| \le \varepsilon \psi^{(k)} & \forall v\in V \mbox{ with } \psi_v \le \psi^{(k)}.
	\end{array}
	\right.
	\end{equation}
\end{restatable}

In terms of time complexity, for individual centrality, lines~\ref{line:estnew1}
	and~\ref{line:estnew2} take constant time for each $u\in R_v$, so it has the
	same complexity as the algorithm in \cite{CT17}.
For Shapley centrality, the following lemma shows that the computation of
	$\phi^{\Shapley}_u(g(d_{R_v}(\cdot,v)))$ for all $v$ is $O(|R_v|)$,
	same as the complexity of generating $R_v$, so it will not add complexity
	to the overall running time.
Suppose $R_v$ has $\Delta$ levels in total (i.e., $\Delta = \max \{d_{R_v}(u', v) \mid u' \in R_v\}$), and let $s_i = |\{u' \mid u' \in R_v,\, d_{R_v}(u', v) \ge i\}|$.
	
\begin{restatable}{mylem}{ShapleyCompute}
%
For any function $g:\R_\infty \rightarrow \R$ with $g(\infty) = 0$,
\begin{align*}
    \phi^{\Shapley}_u(g(d_{R_v}(\cdot,v))) = \frac{1}{|R_v|} g(k) + \frac{1}{|R_v|!} \sum_{k < i \le \Delta} (g(k) - g(i)) \left(\sum_{0 \le j \le s_i} \frac{s_i!}{j!} - \sum_{0 \le j \le s_{i + 1}} \frac{s_{i + 1}!}{j!}\right).
\end{align*}
In $O(|R_v|)$ time we can compute this value for all nodes in $R_v$ (assuming infinite precision).
For degree centrality, $\phi^{\Shapley}_u(g^{\Degree}(d_{R_v}(\cdot,v))) = 
1/|\{w \mid d_{R_v}(w,v) = 1\} |$ if $d_{R_v}(u,v) = 1$, and otherwise it is $0$.
For reachability centrality, $\phi^{\Shapley}_u(g^{\Reachability}(d_{R_v}(\cdot,v))) = 
1/|R_v|$.
\end{restatable}

Therefore, the time complexity follows \cite{CT17}:
\begin{restatable}{mythm}{algotime} \label{thm:algotime}
Under the assumption that sampling a triggering set $T(v)$ takes time
	at most $O(|N^-(v)|)$ time,  and the condition
	$\ell \ge (\log_2 k - \log_2 \log_2 n)/\log_2 n$,
	the expected running time 
	of {\ICERR} is $O(\ell (m+n) \log n \cdot \E[\sigma(\tilde{\bv})]/ (\psi^{(k)} \varepsilon^2))$,
	where $\E[\sigma(\tilde{\bv})]$ is the expected influence spread of a random 
	node $\tilde{\bv}$ drawn from $V$ with probability proportional to the 
	in-degree of $\tilde{\bv}$.
\end{restatable}
Theorems~\ref{thm:cenest} and~\ref{thm:algotime} together show that our algorithm
	{\ICERR} provides a framework to efficiently estimate 
	all individual and Shapley centralities
	in the family of influence-based stochastic sphere-of-influence centralities.
We further remark that, although algorithm {\ICERR} is shown for computing individual
	centralities and Shapley centralities, it can be easily adapted to computing group
	centralities as well.
Of course, a group centrality has $2^n$ values, so it is not feasible to list all of them.
But if we consider that the algorithm is to estimate $n$ group centrality values for
	$n$ given sets, then we only need to replace $\est_u$ with $\est_S$ for every $S$
	in the input, and change the lines corresponding to
	individual centrality (lines~\ref{line:estnew1} and~\ref{line:estnew2}) to
	``for each $S$ in the input,  $\est_{S} = \est_{S} + g(d_{R_v}(S,v))$''.
This change is enough for estimating $n$ group centrality values.

\section{Future Work}

Many topics concerning the interaction between network centralities and influence dynamics
	can be further explored.
One open question is how to extend other centralities that are not covered by
	sphere-of-influence to influence-based centralities.
For example, betweenness centrality of a node $v$ is determined not only by the distance
	from $v$ to other nodes, but by all-pair distances, while PageRank and other
	eigenvalue centralities are determined by the entire graph structure.
Therefore, one may need to capture further aspects of the influence propagation to provide
	natural extensions to these graph-theoretical centralities.
Another open question is how to characterize centrality for a class of influence profiles, 
	e.g.\ all submodular influence profiles, all triggering models, etc.
Empirical comparisons of different influence-based centralities, as well as
	studying the applications that could utilize influence-based centralities, 
	are all interesting and important topics worth further investigation. 


\clearpage

\bibliographystyle{plain}
\bibliography{biblio}

\clearpage

\appendix

\section*{Appendix}

\section{Omitted Proofs}

For convenience, we restate the theorems and lemmas in this section before the proofs.



\snireachability*
\begin{proof}
For any $v \in V$,
\begin{align*}
    \SNI_v(\I) & = \sigma_{\I}(v) \\
    & = \E_{(S_1, \ldots, S_{n-1})\sim P_{\I}} \sum_{u \in V} \mathbb{I}[d_u \le \infty] \\
    & = \E_{(S_1, \ldots, S_{n-1})\sim P_{\I}} f^{\Reachability}(\vec{d}(\{v\}, S_1, \ldots, S_{n-1})) \\
    & = \psi^{\SNSSoI}[f^{\Reachability}](\I).
\end{align*}
\end{proof}

\conformdistfunc*
\begin{proof}
For any directed graph $G=(V,E)$ and any set $v \in V$, 
let $(\{v\}, S^v_1, \ldots, S^v_{n-1})$ be the BFS sequence starting from $v$ in $G$.
Then we have
\begin{align*}
\psi[f]_v(\I_G) & = \E_{(S_1, \ldots, S_{n-1})\sim P_{\I_G}(\{v\})}[f(\vec{d}(\{v\}, S_1, \ldots, S_{n-1}))] \\
& = f(\vec{d}(\{v\}, S^v_1, \ldots, S^v_{n-1})) \\
& = f(\vec{d}_G(v)) = \mu[f]_v(G).
\end{align*}
Thus the lemma holds.
\end{proof}

\layeredgraphbasis*
\begin{proof}
    All we have to show is independence since ${|\cL| = M}$. Suppose not, i.e., there is a nontrivial group of $\{\lambda(R_0, \dots, R_t)\}$ such that 
    \[
        \sum_{R_0, \dots, R_t:\,\I(R_0, \dots, R_t) \in \mathcal{L}} \lambda(R_0, \dots, R_t) P_{\I(R_0, \dots, R_t)} = \mathbf{0}.
    \]
    Let $\I(R_0^*, \dots, R_{t^*}^*)$ be a layered-graph instance
    \begin{itemize}
        \item Such that $\lambda(R_0^*, \dots, R_{t^*}^*) \ne 0$;
        \item Among those satisfying the condition above, with the largest number of layers (i.e.\ $t^*$);
        \item Among those satisfying the conditions above, with the largest number of vertices in the first layer (i.e.\ $|R_0^*|$).
    \end{itemize}
    Note that fixing the seed set, the propagation on a layered-graph instance is deterministic. That is, there is exactly one cascading sequence with the fixed seed set, which happens with probability $1$. Let $\mathit{Seq}_{\I(R_0, \dots, R_t)}(S_0)$ be the unique 
    BFS sequence which happens on $\I(R_0, \dots, R_t)$ with seed set $S_0$. We show that
    \begin{align*}
        \sum_{\emptyset \ne S_0 \subseteq R_0^*} (-1)^{1 + |S_0|} \sum_{R_0, \dots, R_t:\,\I(R_0, \dots, R_t) \in \mathcal{L}} \lambda(R_0, \dots, R_t) P_{\I(R_0, \dots, R_t)}(\mathit{Seq}_{\I(R_0^*, \dots, R_{t^*}^*)}(S_0)) \ne 0,
    \end{align*}
    which contradicts the assumption of non-independence and thereby concludes the proof.
    
    We now compute the left hand side of the above formula.

    \begin{align*}
        & \sum_{\emptyset \ne S_0 \subseteq R_0^*} (-1)^{1 + |S_0|} \sum_{R_0, \dots, R_t:\,\I(R_0, \dots, R_t) \in \mathcal{L}} \lambda(R_0, \dots, R_t) P_{\I(R_0, \dots, R_t)}(\mathit{Seq}_{\I(R_0^*, \dots, R_{t^*}^*)}(S_0)) \\
        =\ & \sum_{R_0, \dots, R_t:\, \I(R_0, \dots, R_t) \in \mathcal{L}} \lambda(R_0, \dots, R_t) \sum_{\emptyset \ne S_0 \subseteq R_0^*} (-1)^{1 + |S_0|} P_{\I(R_0, \dots, R_t)} (\mathit{Seq}_{\I(R_0^*, \dots, R_{t^*}^*)}(S_0)) \\
        =\ & \sum_{t < t^*, R_0, \dots, R_t:\, \I(R_0, \dots, R_t) \in \mathcal{L}} \lambda(R_0, \dots, R_t) \sum_{\emptyset \ne S_0 \subseteq R_0^*} (-1)^{1 + |S_0|} P_{\I(R_0, \dots, R_t)} (\mathit{Seq}_{\I(R_0^*, \dots, R_{t^*}^*)}(S_0)) \\
        +\ & \sum_{t > t^*, R_0, \dots, R_t:\, \I(R_0, \dots, R_t) \in \mathcal{L}} \lambda(R_0, \dots, R_t) \sum_{\emptyset \ne S_0 \subseteq R_0^*} (-1)^{1 + |S_0|} P_{\I(R_0, \dots, R_t)} (\mathit{Seq}_{\I(R_0^*, \dots, R_{t^*}^*)}(S_0)) \\
        +\ & \sum_{R_0, \dots, R_{t^*}:\, \I(R_0, \dots, R_{t^*}) \in \mathcal{L}} \lambda(R_0, \dots, R_{t^*}) \sum_{\emptyset \ne S_0 \subseteq R_0^*} (-1)^{1 + |S_0|} P_{\I(R_0, \dots, R_{t^*})} (\mathit{Seq}_{\I(R_0^*, \dots, R_{t^*}^*)}(S_0)) \\
        =\ & \sum_{R_0, \dots, R_{t^*}:\, \I(R_0, \dots, R_{t^*}) \in \mathcal{L}} \lambda(R_0, \dots, R_{t^*}) \sum_{\emptyset \ne S_0 \subseteq R_0^*} (-1)^{1 + |S_0|} P_{\I(R_0, \dots, R_{t^*})} (\mathit{Seq}_{\I(R_0^*, \dots, R_{t^*}^*)}(S_0)) \\
        =\ & \sum_{R_0:\, R_0 \cap R_0^* \ne \emptyset,\, \I(R_0, R_1^*, \dots, R_{t^*}) \in \mathcal{L}} \lambda(R_0, R_1^*, \dots, R_{t^*}) \\
        \times\ & \sum_{\emptyset \ne S_0 \subseteq R_0^*} (-1)^{1 + |S_0|} P_{\I(R_0, R_1^*, \dots, R_{t^*})} (\mathit{Seq}_{\I(R_0^*, \dots, R_{t^*}^*)}(S_0)).
    \end{align*}
    
    Now consider the summand in the last line of the above equation. For $R_0 \ne R_0^*$,
    \begin{align*}
        & \sum_{\emptyset \ne S_0 \subseteq R_0^*} (-1)^{1 + |S_0|} P_{\I(R_0, R_1^*, \dots, R_{t^*}^*)}(\mathit{Seq}_{\I(R_0^*, \dots, R_{t^*}^*)}(S_0)) \\
        =\ & \sum_{\emptyset \ne X \subseteq R_0^* \cap R_0} \sum_{Y \subseteq R_0^* \setminus R_0} (-1)^{1 + |X| + |Y|} P_{\I(R_0, R_1^*, \dots, R_{t^*}^*)}(\mathit{Seq}_{\I(R_0^*, \dots, R_{t^*}^*)}(X \cup Y)) \\
        +\ & \sum_{\emptyset \ne Y \subseteq R_0^* \setminus R_0} (-1)^{1 + |Y|} P_{\I(R_0, R_1^*, \dots, R_{t^*}^*)}(\mathit{Seq}_{\I(R_0^*, \dots, R_{t^*}^*)}(Y)) \\
        =\ & \sum_{\emptyset \ne X \subseteq R_0^* \cap R_0} (-1)^{1 + |X|} \sum_{Y \subseteq R_0^* \setminus R_0} (-1)^{|Y|}\\
        =\ & \sum_{\emptyset \ne X \subseteq R_0^* \setminus R_0} (-1)^{1 + |Y|} \times 0 \\
        =\ & 0.
    \end{align*}
    And for $R_0 = R_0^*$,
    \begin{align*}
        & \sum_{\emptyset \ne S_0 \subseteq R_0^*} (-1)^{1 + |S_0|} P_{\I(R_0^*, R_1^*, \dots, R_{t^*}^*)}(\mathit{Seq}_{\I(R_0^*, \dots, R_{t^*}^*)}(S_0)) = \sum_{\emptyset \ne S_0 \subseteq R_0^*} (-1)^{1 + |S_0|} = 1.
    \end{align*}
    Plugging these in, we get
    \begin{align*}
        & \sum_{\emptyset \ne S_0 \subseteq R_0^*} (-1)^{1 + |S_0|} \sum_{R_0, \dots, R_t:\,\I(R_0, \dots, R_t) \in \mathcal{L}} \lambda(R_0, \dots, R_t) P_{\I(R_0, \dots, R_t)}(\mathit{Seq}_{\I(R_0^*, \dots, R_{t^*}^*)}(S_0)) \\
        =\ & \sum_{R_0:\, R_0 \cap R_0^* \ne \emptyset,\, \I(R_0, R_1^*, \dots, R_{t^*}^*) \in \mathcal{L}} \lambda(R_0, R_1^*, \dots, R_{t^*}^*) \sum_{\emptyset \ne S_0 \subseteq R_0^*} (-1)^{1 + |S_0|} P_{\I(R_0, R_1^*, \dots, R_{t^*}^*)} (\mathit{Seq}_{\I(R_0^*, \dots, R_{t^*}^*)}(S_0)) \\
        =\ & \sum_{R_0:\, R_0 \cap R_0^* \ne \emptyset,\, R_0 \ne R_0^*,\, \I(R_0, R_1^*, \dots, R_{t^*}^*) \in \mathcal{L}} \lambda(R_0, R_1^*, \dots, R_{t^*}^*) \times 0 + \lambda(R_0^*, R_1^*, \dots, R_{t^*}^*) \times 1 \\
        =\ & \lambda(R_0^*, \dots, R_{t^*}^*) \\
        \ne\ & 0.
    \end{align*}
\end{proof}

\distfuncanonymous*
\begin{proof}
\begin{align*}
& \psi^{\Individual}[f]_{\pi(v)}(\pi(\I)) & \\
=\ & \E_{(S_1, \ldots, S_{n-1})\sim P_{\pi(\I)}(\{\pi(v)\})}[f(\vec{d}(\{\pi(v)\}, \pi(S_1), \ldots, \pi(S_{n-1})))] \\
=\ & \E_{(S_1, \ldots, S_{n-1})\sim P_{\I}(\{v\})}[f(\vec{d}(\{v\}, S_1, \ldots, S_{n-1}))]\\
=\ & \psi^{\Individual}[f]_v(\I).
\end{align*}
\end{proof}

\distfuncbayesian*
\begin{proof} 
	Suppose $\psi[f]$ is an influence-based distance-function centrality measure.
	Let instances $\I$, $\I_1$, $\I_2$ be that
	$ P_{\I} = \alpha P_{\I_1} + (1 - \alpha) P_{\I_2}$.
	Then for every sequence $(S_1, \ldots, S_{n-1})$ drawn from distribution $P_{\I}(\{v\})$, 
	it is equivalently drawn with probability $\alpha$ from $P_{\I_1}(\{v\})$, and
	with probability $1-\alpha$ from $P_{\I_2}(\{v\})$.
	Therefore, 
	\begin{align*}
	\psi[f]_v(\I) & = \E_{(S_1, \ldots, S_{n-1})\sim P_{\I}(\{v\})}[f(\vec{d}(\{v\}, S_1, \ldots, S_{n-1}))] \\
	& = \alpha \cdot \E_{(S_1, \ldots, S_{n-1})\sim P_{\I_1}(\{v\})}[f(\vec{d}(\{v\}, S_1, \ldots, S_{n-1}))] \\
    & + (1-\alpha) \cdot \E_{(S_1, \ldots, S_{n-1})\sim P_{\I_2}(\{v\})}[f(\vec{d}(\{v\}, S_1, \ldots, S_{n-1}))] \\
	& = \alpha \cdot \psi[f]_v(\I_1) + (1-\alpha) \cdot \psi[f]_v(\I_2).
	\end{align*}
	Thus the lemma holds.
\end{proof}

The proof of Theorem~\ref{thm:unique_determination} relies on a general lemma about linear mapping
	as given in \cite{CT17}, as restated below.
\begin{lemma}[Lemma 11 of \cite{CT17}] \label{lem:linearmap}
	Let $\psi$ be a mapping from a convex set $D \subseteq \R^M$ to $\R^n$ satisfying
	that for any vectors $\vv_1, \vv_2, \ldots, \vv_s \in D$, 
	for any $\alpha_1, \alpha_2, \ldots, \alpha_s \ge 0$ and $\sum_{i=1}^s \alpha_i = 1$,
	$\psi(\sum_{i=1}^{s} \alpha_i \cdot \vv_i )=  \sum_{i=1}^{s} \alpha_i \cdot \psi(\vv_i)$.
	Suppose that $D$ contains a set of linearly independent basis vectors of $\R^M$,
	$\{\vb_1, \vb_2, \ldots, \vb_M\}$ and also vector $\vzero$.
	Then for any $\vv\in D$, which can be represented as $\vv = \sum_{i=1}^M \lambda_i \cdot \vb_i$
	for some $\lambda_1, \lambda_2, \ldots, \lambda_M \in \R$,
	we have 
	\[
	\psi(\vv) = \psi\left(\sum_{i=1}^M \lambda_i \cdot \vb_i\right) 
	= \sum_{i=1}^M \lambda_i \cdot \psi(\vb_i) + \left(1-\sum_{i=1}^M \lambda_i \right) \cdot \psi(\vzero).
	\]
\end{lemma}

\uniquedetermine*
\begin{proof}[Proof of Theorem~\ref{thm:unique_determination}]
By the definition of the null influence instance $\I^N$ (same as the trivial layered-graph instance
	$\I_V(V)$), we can see that the vector representation
	of the null influence instance is the all $0$ vector, because the entries corresponding to
	$P_{\I^N}(S_0, S_0, \ldots, S_0)$ are not included in the vector by definition.
Then by Theorem~\ref{thm:basis} and the Lemma~\ref{lem:linearmap}, we know that for any
	Bayesian centrality measure $\psi$, its value on any influence instance $\I$,
	$\psi(\I)$, can be represented as a linear combination of the $\I$'s values on layered-graph
	instances (including the null instance).
Thus, the theorem holds.
\end{proof}

%

\IndividualCharacterization*
\begin{proof}
The theorem follows directly from Propositions~\ref{prop:distfuncbayesian}, \ref{prop:distfuncanonymous} and \ref{prop:conformdistfunc}.
\end{proof}

\GroupCharacterization*
\begin{proof}
The Bayesian part of the proof is essentially the same as the proof of
    Proposition~\ref{prop:distfuncbayesian}, with subset $S$ replacing node $v$.
The Anonymity part of the proof is essentially the same as the proof of
    Proposition~\ref{prop:distfuncanonymous}, with subset $S$ replacing node $v$.
The conformity part of the proof is essentially the same as the proof of
    Proposition~\ref{prop:conformdistfunc}, with subset $S$ replacing node $v$.
\end{proof}

\ShapleyCharacterization*
\begin{proof}
    We first show that $\psi^{\Shapley}[f]$ is Bayesian.
    Let  $\I$, $\I_1$, $\I_2$ be three influence instances with the same vertex set $V$,
    and $\alpha \in [0, 1]$, where $ P_{\I} = \alpha P_{\I_1} + (1 - \alpha) P_{\I_2}$.
    Then for every node $v\in V$, we have
    \begin{align}
    \psi_v^{\Shapley}(\I) &= \phi_v^{\Shapley}(\psi^{\Group}(\I)) \nonumber \\
    & = \phi_v^{\Shapley}(\psi^{\Group}(\alpha P_{\I_1} + (1 - \alpha) P_{\I_2})) \nonumber\\
    & = \phi_v^{\Shapley}(\alpha \psi^{\Group}(\I_1) + (1-\alpha) \psi^{\Group}(\I_2) ) \nonumber\\
    & = \alpha \phi_v^{\Shapley}(\psi^{\Group}(\I_1)) + (1-\alpha) \phi_v^{\Shapley}(\psi^{\Group}(\I_2)) 
    \label{eq:shapleylinear}\\
    & = \alpha \psi_v^{\Shapley}(\I_1) + (1-\alpha) \psi_v^{\Shapley}(\I_2), \nonumber
    \end{align} 
    where Eq.~\eqref{eq:shapleylinear} is due to the linearity of the Shapley value, which is
    easy to verify by the following derivations:
    \begin{align*}
    \phi_v^{\Shapley}(\alpha\tau_1 + \beta \tau_2) & = \E_{\pi\sim \Pi}[(\alpha\tau_1 + \beta \tau_2)(S_{\pi,v} \cup \{v\}) - 
    (\alpha\tau_1 + \beta \tau_2)(S_{\pi,v})] \\
    & = \alpha \E_{\pi\sim \Pi}[\tau_1(S_{\pi,v} \cup \{v\}) - \tau_1(S_{\pi,v})] + \beta \E_{\pi\sim \Pi}[\tau_2(S_{\pi,v} \cup \{v\})- \tau_2(S_{\pi,v})] \\
    & = \alpha \phi_v^{\Shapley}(\tau_1) + \beta \phi_v^{\Shapley}(\tau_2).
    \end{align*}
    Now we show that $\psi^{\Shapley}[f]$ conforms with $\mu^{\Shapley}[f]$.
    For any directed graph $G=(V,E)$ and any node $v\in V$, let
    $(\{v\}, S^v_1, \ldots, S^v_{n-1})$ be the BFS sequence starting from $v$ in graph $G$.
    We have
    \begin{align}
    \psi^{\Shapley}[f]_v(\I_G) & = \phi^{\Shapley}_v(\psi^{\Group}[f](\I_G))  \nonumber \\
    & = \E_{\pi\sim \Pi}[\psi^{\Group}[f]_{S_{\pi,v} \cup \{v\}}(\I_G) - \psi^{\Group}[f]_{S_{\pi,v}}(\I_G)] \nonumber \\
    & = \E_{\pi\sim \Pi}[\mu^{\Group}[f]_{S_{\pi,v} \cup \{v\}}(G) - \mu^{\Group}[f]_{S_{\pi,v}}(G)] \label{eq:groupconform} \\
    & = \phi^{\Shapley}_v(\mu^{\Group}[f](G)) \nonumber \\
    & = \mu^{\Shapley}[f](G), \nonumber
    \end{align}
    where Eq.~\eqref{eq:groupconform} is because influence-based group centrality
    $\psi^{\Group}[f]$ conforms with the structure-based group centrality
    $\mu^{\Group}[f]$ (Theorem~\ref{thm:GroupCharacterization}). Anonymity follows from anonymity of group centralities (Theorem~\ref{thm:GroupCharacterization}).
    Uniqueness then follows from Theorem~\ref{thm:unique_determination}.
\end{proof}

\begin{lemma} \label{lem:distequivalence}
	For fixed nodes $u , w\in V$, suppose we generate a random RR set $R_w$ rooted at $w$,
	according to a triggering model instance $\I$.
	Then we have 
	\[
	\E[g(d_{R_w}(u,w))] = 
	\E_{(S_1, \ldots, S_{n-1})\sim P_{\I}(\{u\})}[g(d_w(\{u\}, S_1, \ldots, S_{n-1}))].
	\]
\end{lemma}
\begin{proof}
	We know that the triggering model is equivalent to 
	the following live-edge graph model~\cite{kempe03}:
	for every node $v\in V$, sample its triggering set $T(v)$ and add edges
	$(u,v)$ to a live-edge graph $L$ for all $u\in T(v)$ (these edges are called
	live edges).
	Then the diffusion from a seed set $S$ is the same as the BFS propagation in $L$ from $S$.
	Since reverse simulation for generating RR set $R_w$ also do the same sampling of the
	triggering sets, we can couple the reverse simulation process with the forward
	propagation by fixing a live-edge graph $L$.
	For a fixed live-edge graph $L$, the subgraph $G_{R_w}$ generated by reverse simulation
	from the fixed root $w$ is simply the induced subgraph of $L$ induced by all nodes
	that can reach $w$ in $L$.
	Thus $d_{R_w}(u,w)$ is the fixed distance from $u$ to $w$ in $L$, namely $d_L(u,w)$.
	On the other hand, with the fixed $L$, the cascading sequence starting from $u$
	is the fixed BFS sequence starting from $u$ in $L$.
	Then in this BFS sequence $d_w(\{u\}, S_1, \ldots, S_{n-1})$ is the distance
	from $u$ to $w$ in this sequence, which is the same as the distance from $u$ to $w$
	in $L$.
	Therefore, $d_w(\{u\}, S_1, \ldots, S_{n-1}) = d_{R_w}(u,w)$ for fixed live-edge graph $L$.
	We can then vary $L$ according to its distribution, and obtain
	\[
	\E[g(d_{R_w}(u,w))] = 
	\E_{(S_1, \ldots, S_{n-1})\sim P_{\I}(\{u\})}[g(d_w(\{u\}, S_1, \ldots, S_{n-1}))].
	\]
\end{proof}

\RRsetCrucial*
\begin{proof}
	Consider first the influence-based distance-function centrality.
	We have
	\begin{align}
	n\cdot \E[g(d_{R_v}(u,v))] & = n \cdot \sum_{w\in V} \Pr\{w = v\} \E[g(d_{R_v}(u,v)) \mid w = v] \nonumber \\
	& = n \cdot \sum_{w\in V} \frac{1}{n} \E[g(d_{R_w}(u,w))] \nonumber \\
	& = \sum_{w\in V} 
	\E_{(S_1, \ldots, S_{n-1})\sim P_{\I}(\{u\})}[g(d_w(\{u\}, S_1, \ldots, S_{n-1}))] 
	\label{eq:distequv} \\
	& = \E_{(S_1, \ldots, S_{n-1})\sim P_{\I}(\{u\})}\left[ \sum_{w\in V} 
	g(d_w(\{u\}, S_1, \ldots, S_{n-1})) \right] \nonumber \\
	& = \E_{(S_1, \ldots, S_{n-1})\sim P_{\I}(\{u\})}[f(\vec{d}(\{u\}, S_1, \ldots, S_{n-1}))] 
	\nonumber \\
	& = \psi[f]_u(\I), \nonumber
	\end{align}	
	where Eq.~\eqref{eq:distequv} is by Lemma~\ref{lem:distequivalence}.
	
	Next consider the influence-based distance-function Shapley centrality.
	\begin{align}
	n\cdot \E[\phi^{\Shapley}_u(g(d_{R_v}(\cdot,v)))] & = n \cdot \sum_{w\in V} \Pr\{w = v\} \E[\phi^{\Shapley}_u(g(d_{R_v}(\cdot,v))) \mid w = v] \nonumber \\
	& = n \cdot \sum_{w\in V} \frac{1}{n} \E[\phi^{\Shapley}_u(g(d_{R_w}(\cdot,w)))] \nonumber \\
	& = \phi^{\Shapley}_u \left(\sum_{w\in V} \E[g(d_{R_w}(\cdot,w))  ] \right) 
	\label{eq:ShapleyLinear} \\
	& = \phi^{\Shapley}_u \left( \psi^{\Group}[f](\I) \right) 
	\label{eq:dist2group} \\
	& = \psi^{\Shapley}[f]_u(\I), \nonumber
	\end{align}	
	where Eq.~\eqref{eq:ShapleyLinear} is by the linearity of the Shapley value, as already
	argued in the proof of Theorem~\ref{thm:GroupCharacterization}, and
	Eq.~\eqref{eq:dist2group} follows the similar derivation step as in the case
	of individual centrality above.
\end{proof}

\thmAlgo*
\begin{proof}
The proof follows exactly the same proof structure as the proof in \cite{CT17}.
All we need to change is to use our crucial lemma connecting RR sets with 
	centrality measures (Lemma~\ref{lem:RRsetCrucial}) to replace the corresponding
	lemma (Lemma 23) in \cite{CT17}.
\end{proof}

\ShapleyCompute*
\begin{proof}
    We prove only for general $g$.
    For $u$ at the $k$-th level,
    \begin{align*}
        & \phi^{\Shapley}_u(g(d_{R_v}(\cdot,v))) \\
        =\ & \frac{1}{|R_v|!} \sum_\pi (g(d_{R_v}(S_{\pi, u} \cup \{u\}, v)) - g(d_{R_v}(S_{\pi, u}, v))) \\
        =\ & \frac{1}{|R_v|} g(k) + \frac{1}{|R_v|!} \sum_\pi \sum_{k < i \le \Delta} \mathbb{I}[d_{R_v}(S_{\pi, u}, v) = i] (g(k) - g(i)) \\
        =\ & \frac{1}{|R_v|} g(k) + \frac{1}{|R_v|!} \sum_{k < i \le \Delta} (g(k) - g(i)) \sum_{1 \le j \le |R_v|} \left[\binom{s_i}{j - 1}(j - 1)! - \binom{s_{i + 1}}{j - 1}(j - 1)!\right] \\
        =\ & \frac{1}{|R_v|} g(k) + \frac{1}{|R_v|!} \sum_{k < i \le \Delta} (g(k) - g(i)) \left(\sum_{1 \le j \le s_i + 1} \frac{s_i!}{(s_i - j + 1)!} - \sum_{1 \le j \le s_{i + 1} + 1} \frac{s_{i + 1}!}{(s_{i + 1} - j + 1)!}\right) \\
        =\ & \frac{1}{|R_v|} g(k) + \frac{1}{|R_v|!} \sum_{k < i \le \Delta} (g(k) - g(i)) \left(\sum_{0 \le j \le s_i} \frac{s_i!}{j!} - \sum_{0 \le j \le s_{i + 1}} \frac{s_{i + 1}!}{j!}\right).
    \end{align*}

    One possible way to compute the value above is:
    \begin{enumerate}
        \item Compute $x!$ and $\sum_{0 \le i \le x} \frac{1}{i!}$ for all $x \in [|R_v|]$ in time $O(|R_v|)$.
        \item Now the term $\sum_{0 \le j \le s_i} \frac{s_i!}{j!}$ can be computed in additional constant time for any $i$. We can compute 
        \[
            \frac{1}{|R_v|!} \sum_{k < i \le \Delta} \left(g(k) - g(i)\right) \left(\sum_{0 \le j \le s_i} \frac{s_i!}{j!} - \sum_{0 \le j \le s_{i + 1}} \frac{s_{i + 1}!}{j!}\right)
        \]
        for all $0 \le k \le \Delta$ in additional total time $O(\Delta)$ by first computing a suffix sum of
        \[
            g(i) \left(\sum_{0 \le j \le s_i} \frac{s_i!}{j!} - \sum_{0 \le j \le s_{i + 1}} \frac{s_{i + 1}!}{j!}\right).
        \]
        That is,
        \[
            \sum_{k < i \le \Delta} g(i) \left(\sum_{0 \le j \le s_i} \frac{s_i!}{j!} - \sum_{0 \le j \le s_{i + 1}} \frac{s_{i + 1}!}{j!}\right).
        \]
        for all $k \in \{0, \dots, \Delta\}$.
        \item Assign the values to vertices in each level in total time $O(|R_v|)$.
    \end{enumerate}

\end{proof}

\end{document}